\documentclass[12pt,reqno]{amsart}
\usepackage{latexsym,amsmath}
\usepackage{amsthm}
\usepackage{amssymb}
\usepackage{amsthm}
\usepackage{epsfig}
\usepackage[left=3cm,top=3cm,right=3cm,bottom=3cm]{geometry}

\begin{document}
\newtheorem{theorem}{Theorem}[section]
\newtheorem{cor}[theorem]{Corollary}
\newtheorem{lemma}[theorem]{Lemma}
\newtheorem{fact}[theorem]{Fact}
\newtheorem{claim}[theorem]{Claim}
\newtheorem{definition}[theorem]{Definition}
\theoremstyle{definition}
\newtheorem{example}[theorem]{Example}
\newtheorem{remark}[theorem]{Remark}
\newcommand\eps{\varepsilon}
\newcommand{\pp}{\mathcal P}
\newcommand{\E}{\mathbb E}
\newcommand{\Var}{{\rm Var}}
\newcommand{\Prob}{\mathbb{P}}
\newcommand{\w}{\mathbf{w}}
\newcommand{\vol}{{\rm vol}}
\newcommand{\N}{{\mathbb N}}
\newcommand{\R}{{\mathbb R}}
\newcommand{\wep}{\emph{wep}}
\newcommand{\eqn}[1]{(\ref{#1})}

\newcommand\qedc{\hfill $\Diamond$ \smallskip}
\renewcommand{\labelenumi}{(\roman{enumi})}

\def\bibitemx#1{\bibitem{#1}\marginpar{$\otimes$}}
\def\co{~\marginpar{\fbox{,}}}
\def\wck#1 {\underline{#1}~\marginpar{\fbox{#1} {\tiny ?}}}
\def\silent#1\par{\par}
\def\text#1{\quad\mbox{#1}\quad}

\def\quest#1{\footnote{$\rightarrow$ {\sf #1}}}

\makeatletter
\renewcommand{\@seccntformat}[1]{\@nameuse{the#1}.\quad}
\makeatother

\title{Geometric protean graphs}

\author{Anthony Bonato}
\address{Department of Mathematics, Ryerson University, Toronto, ON, Canada M5B 2K3}
\email{\tt abonato@ryerson.ca}
\author{Jeannette Janssen}
\address{Department of Mathematics and Statistics, Dalhousie University, Halifax, NS, Canada B3H 3J5}
\email{\tt janssen@mathstat.dal.ca}
\author{Pawe{\l} Pra{\l}at}
\address{Department of Mathematics, West Virginia University, Morgantown, WV 26506-6310, USA}
\email{\tt pralat@math.wvu.edu}

\thanks{The authors gratefully acknowledge support from NSERC and MITACS. The present article is the full version of an article which appeared in the Proceedings of the 7th Workshop on Algorithms and Models for the Web-Graph (WAW2010)~\cite{waw2010}}

\maketitle

\begin{abstract}
We study the link structure of on-line social networks (OSNs), and introduce a new model for such networks which may help infer their hidden underlying reality. In the geo-protean (GEO-P) model for OSNs nodes are identified with points in Euclidean space, and edges are stochastically generated by a mixture of the relative distance of nodes and a ranking function. With high probability, the GEO-P model generates graphs satisfying many observed properties of OSNs, such as power law degree distributions, the small world property, densification power law, and bad spectral expansion. We introduce the dimension of an OSN based on our model, and examine this new parameter using actual OSN data. We discuss how the geo-protean model may eventually be used as a tool to group users with similar attributes using only the link structure of the network.
\end{abstract}

\section{Introduction}

On-line social networking sites such as Facebook, Flickr, LinkedIn, MySpace, and Twitter are examples of large-scale, complex, real-world networks, with an estimated total number of users that equals at least half of all Internet users~\cite{ahn}. We may model an OSN by a graph with nodes representing users and edges corresponding to friendship links. While OSNs gain increasing popularity among the general public, there is a parallel increase in interest in the cataloguing and modelling of their structure, function, and evolution. OSNs supply a vast and historically unprecedented record of large-scale human social interactions over time.

The availability of large-scale social network data has led to numerous studies that revealed emergent topological properties of OSNs. For example, the recent study~\cite{kwak} crawled the entire Twitter site, and studied properties found among the $41.7$ million user profiles and $1.47$ billion social relations.  The next challenge is the design and rigorous analysis of models simulating these properties. Graph models were successful in simulating properties of other complex networks such as the web graph (see the books~\cite{AB,CL} for surveys of such models), and it is thus natural to propose models for OSNs. Few rigorous models for OSNs have been posed and analyzed, and there is no universal consensus of which properties such models should simulate. Notable recent models are those of Kumar et al.~\cite{kumar2006sae}, Lattanzi and Sivakumar~\cite{ls}, and the Iterated Local Transitivity model~\cite{anppc}.

Researchers are now in the enviable position of observing how OSNs evolve over time, and as such, network analysis and models of OSNs typically incorporate time as a parameter. While by no means exhaustive, some of the main observed properties of OSNs include the following.

(i) \emph{Large-scale.} OSNs are examples of complex networks with number nodes (which we write as $n$) often in the millions; further, some users have disproportionately high degrees. For example, some of the nodes of Twitter corresponding to well-known celebrities  have degree over five million.

(ii) \emph{Small world property and shrinking distances.} The small world property, introduced by Watts and Strogatz~\cite{sw}, is a central notion in the study of complex networks (see also~\cite{klein1}). The small world property demands a low diameter of $O(\log n)$, and a higher clustering coefficient than found in a binomial random graph with the same number of nodes and same average degree. Adamic et al.~\cite{ada} provided an early study of an OSN at Stanford University, and found that the network has the small world property. Similar results were found in~\cite{ahn} which studied Cyworld, MySpace, and Orkut, and in~\cite{mislove2007maa} which examined data collected from Flickr, YouTube, LiveJournal, and Orkut. Low diameter (of $6$) and high clustering coefficient were reported in the Twitter by both Java et al.~\cite{twit0} and Kwak et al.~\cite{kwak}. Kumar et al.~\cite{kumar2006sae} reported that in Flickr and Yahoo!360 the diameter actually decreases over time. Similar results were reported for Cyworld in~\cite{ahn}. Well-known models for complex networks such as preferential attachment or copying models have logarithmically growing diameters with time. Various models (see~\cite{les1,les2}) were proposed simulating power law degree distributions and decreasing distances.

(iii) \emph{Densification power law.}  A graph $G$ with $e_t$ edges and $n_t$ nodes satisfies a \emph{densification power law} if there is a constant $a$ in $(1,2)$ such that $e_t$ is proportional to $n_t^a$. In particular, the average degree grows to infinity with the order of the network. In \cite{les1}, densification power laws were reported in several real-world networks such as the physics citation
graph and the internet graph at the level of autonomous systems. Densification was reported in Cyworld \cite{ahn} and has been detected in other OSNs.

(iv) \emph{Power law degree distributions.} In a graph $G$ of order $n,$ let $N_k=N_k(n)$ be the number of nodes of degree $k.$ The degree distribution of $G$ follows a \emph{power law} if $N_k$ is proportional to $k^{-b },$ for a fixed exponent $b>2.$  Power laws were observed over a decade ago in subgraphs sampled from the web graph, and are ubiquitous properties of complex networks (see Chapter~2 of \cite{AB}).  Kumar, Novak, and Tomkins~\cite{kumar2006sae} studied the evolution of Flickr and Yahoo!360, and found that these networks exhibit power-law degree distributions. Power law degree distributions for both the in- and out-degree distributions were documented in Flickr, YouTube, LiveJournal, and Orkut~\cite{mislove2007maa}, as well as in Twitter~\cite{twit0,kwak}.

(vi) \emph{Bad spectral expansion.} Social networks often organize into separate clusters in which the intra-cluster links are significantly higher than the number of inter-cluster links. In particular, social networks contain communities (characteristic of social organization), where tightly knit groups correspond to the clusters~\cite{gn}. As a result, it is reported in~\cite{estrada} that social networks, unlike other complex networks, possess bad spectral expansion properties realized by small gaps between the first and second eigenvalues of their adjacency matrices.

Our main contributions in the present work are twofold: to provide a model---the geo-protean (GEO-P) model---which provably satisfies all six properties above (see Section~\ref{results}; note that, while the model does not generate graphs with shrinking distances, the parameters can be adjusted to give constant diameter), and second, to suggest a reverse engineering approach to OSNs. Given only the link structure of OSNs, we ask whether it is possible to infer the hidden reality of such networks. Can we group users with similar attributes from only the link structure? For instance, a reasonable assumption is that out of the millions of users on a typical OSN, if we could assign the users various attributes such as age, sex, religion, geography, and so on, then we should be able to identify individuals or at least small sets of users by their set of attributes. Thus, if we can infer a set of identifying attributes for each node from the link structure, then we can use this in formation to recognize communities and understand connections between users.

Characterizing  users by a set of attributes leads naturally to a vector-based or geometric approach to OSNs. In geometric graph models, nodes are identified with points in a metric space, and edges are introduced by probabilistic rules that depend on the proximity of the nodes in the space. We envision OSNs as embedded in a \emph{social space}, whose dimensions quantify user traits such as interests or geography; for instance, nodes representing users from the same city or in the same profession would likely be closer in social space. A first step in this direction was given in~\cite{ln}, which introduced a rank-based model in an $m$-dimensional grid for social networks (see also the notion of \emph{social distance} provided in~\cite{watts1}). Such an approach was taken in geometric preferential attachment models of Flaxman et al.~\cite{flax1}, and in the SPA geometric model for the web graph~\cite{abcjp}.

The geo-protean model incorporates a geometric view of OSNs, and also exploits ranking to determine the link structure. Higher ranked nodes are more likely to receive links. A formal description of the model is given in Section~\ref{model}. Results on the model are summarized in Section~\ref{results}. We present a novel approach to OSNs by assigning them a dimension; see the formula~(\ref{mmm}). Given certain OSN statistics (order, power law exponent, average degree, and diameter), we can assign each OSN a dimension based on our model. The dimension of an OSN may be roughly defined as the least integer $m$ such that we can accurately embed the OSN in $m$-dimensional Euclidean space. Full proofs of our results are presented in Section~\ref{proofs}. In the final discussion, we summarize our findings and conjecture on the correct diameter for OSNs.

\section{The GEO-P Model for OSNs}\label{model}

We now present our model for OSNs, which is based on both the notions of embedding the nodes in a metric space (geometric), and a link probability based on a ranking of the nodes (protean). We identify the users of an OSN with points in $m$-dimensional Euclidean space.  Each node has a region of influence, and nodes may be joined with a certain probability if they land within each others region of influence. Nodes are ranked by their popularity from $1$ to $n$, where $n$ is the number of nodes, and $1$ is the highest ranked node. Nodes that are ranked higher  have larger regions of influence, and so are more likely to acquire links over time. For simplicity, we consider only undirected graphs. The number of nodes $n$ is fixed but the model is dynamic: at each time-step, a node is born and one dies. A static number of nodes is more representative of the reality of OSNs, as the number of users in an OSN would typically have a maximum (an absolute maximum arises from roughly the number of users on the internet, not counting multiple accounts). For a discussion of ranking models for complex networks, see~\cite{FFM, HP, JJPP2,TLPP}.

We now formally define the GEO-P model. The model produces a sequence $(G_t :t\ge 0)$ of undirected graphs on $n$ nodes, where $t$ denotes time. We write $G_t=(V_t,E_t).$ There are four parameters: the \emph{attachment strength} $\alpha \in (0,1)$, the \emph{density parameter} $\beta \in (0,1-\alpha)$, the \emph{dimension} $m \in \N$, and the \emph{link probability} $p \in (0,1]$. Each node $v\in V_t$ has rank $r(v,t) \in [n]$ (we use $[n]$ to denote the set $\{1,2,\dots, n\}$). The rank function $r(\cdot,t):V_t\rightarrow [n]$ is a bijection for all $t$, so every node has a unique rank. The highest ranked node has rank equal to 1; the lowest ranked node has rank $n$. The initialization and update of the ranking is done by \emph{random initial rank} (Other ranking schemes may also be used. We use random initial rank for its simplicity.) In particular, the node added at time $t$ obtains an initial rank $R_t$ which is randomly chosen from $[n]$ according to a prescribed distribution. Ranks of all nodes are adjusted accordingly. Formally, for each $v\in V_{t-1}$ that is not deleted at time $t$,
$$
r(v,t)=r(v,t-1)+\delta-\gamma,
$$
where $\delta=1$ if $r(v,t-1)>R_t$ and $0$ otherwise, and $\gamma =1$ if the rank of the node deleted in step $t$ is smaller than $r(v,t-1)$, and $0$ otherwise.

Let $S$ be the unit hypercube in $\R^m$, with the torus metric $d(\cdot,\cdot)$ derived from the $L_\infty$ metric. More precisely, for any two points $x$ and $y$ in $\R^m$, their distance is given by
$$
d(x,y)=\min\{ ||x-y+u||_\infty\,:\,u\in \{-1,0,1\}^m\}.
$$
The torus metric thus ``wraps around" the boundaries of the unit cube, so every point in $S$ is equivalent. The torus metric is chosen so that there are no boundary effects, and altering the metric will not significantly affect the main results.

To initialize the model, let $G_0=(V_0,E_0)$ be any graph on $n$ nodes that are chosen from $S$. We define the \emph{influence region} of node $v$ at time $t\geq 0$, written $R(v,t),$ to be the ball around $v$ with volume
$$
|R(v,t)|= r(v,t)^{-\alpha} n^{-\beta} \;.
$$
For $t\ge 1,$ we form $G_t$ from $G_{t-1}$ according to the following rules.
\begin{enumerate}
   \item Add a new node $v$ that is chosen \emph{uniformly at random} from $S$. Next, independently, for each node $u\in V_{t-1}$ such that $v \in R(u,t-1)$, an edge $vu$ is created with probability $p$. Note that the probability that $u$ receives an edge is proportional to $p\, r(u,t-1)^{-\alpha}.$ The negative exponent guarantees that nodes with higher ranks ($r(u,t-1)$ close to 1) are more likely to receive new edges than lower ranks.
   \item Choose uniformly at random a node $u \in V_{t-1}$, delete $u$ and all edges incident to $u$.
   \item Vertex $v$ obtains an initial rank $r(v,t)=R_t$ which is randomly chosen from $[n]$ according to a prescribed distribution.
   \item Update the ranking function $r(\cdot, t) : V_t \to [n].$
 \end{enumerate}

Since the process is an ergodic Markov chain, it will converge to a stationary distribution. (See~\cite{LPW} for more on Markov chains.) The random graph corresponding to this distribution with given parameters $\alpha, \beta, m, p$ is called the \emph{geo-protean} graph (or \emph{GEO-P} model), and is written {GEO-P}$(\alpha,\beta,m,p).$ The coupon collector problem can give us insight into when the stationary state will be reached. Namely, let $L = n(\log n + O(\omega(n)))$ where $\omega(n)$ is any function tending to infinity with $n$. It is a well-known result that, with probability tending to 1 as $n$ tends to infinity, after $L$ steps all original vertices will be deleted.

See Figure~\ref{fig:1} for a simulation of the model in the unit square.

\begin{figure}[h!]
\begin{center}
\epsfig{figure=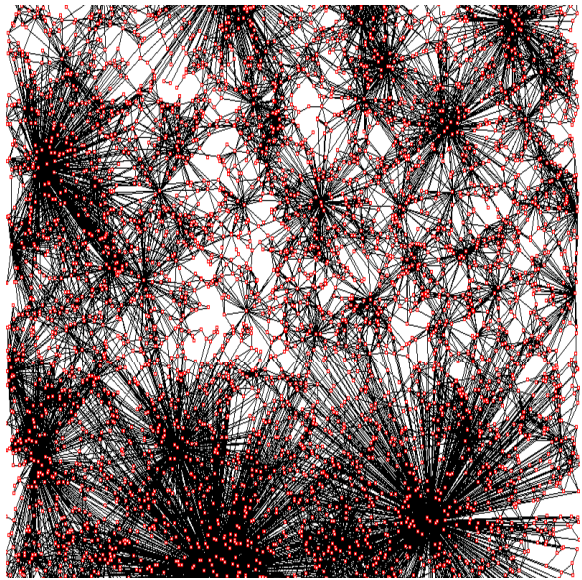, width=3in, height=3in}
\caption{A simulation of the GEO-P model, with $n=5,000$, $\alpha=0.7,$ $\beta =0.15,$ $m=2,$ and $p=0.9.$}\label{fig:1}
\end{center}
\end{figure}

\section{Results and Dimension}\label{results}

\subsection{Results}

We now state the main theoretical results we discovered for the geo-protean model, with proofs supplied in the next section. The model generates with high probability graphs satisfying each of the properties (i) to (iv) we discussed in the introduction. Proofs are presented in Section~\ref{proofs}. Throughout, we will use the stronger notion of \emph{wep} in favour of the more commonly used \emph{aas}, since it simplifies some of our proofs. We say that an event holds \emph{with extreme probability} (\emph{wep}), if it holds with probability at least $1-\exp (-\Theta(\log^2 n))$ as $n \to \infty$.  Thus, if we consider a polynomial number of events that each holds \emph{wep}, then \emph{wep} all events hold.

\bigskip

Let $N_k=N_k(n,p,\alpha,\beta)$ denote the number of nodes of degree $k,$ and $N_{\ge k}=\sum_{l\ge k} N_l$. The following theorem demonstrates that the geo-protean model generates power law graphs with exponent
\begin{equation}\label{bee}
b =1 + 1/\alpha.
\end{equation}
Note that the variables $N_{\ge k}$ represent the cumulative degree distribution, so the degree distribution of these variables has power law exponent $1/\alpha.$

\begin{theorem}\label{powerlaw}
Let $\alpha \in (0,1)$, $\beta \in (0,1-\alpha)$, $m \in \N$, $p\in (0,1]$, and
$$
n^{1-\alpha-\beta} \log^{1/2} n \le k \le n^{1-\alpha/2-\beta} \log^{-2 \alpha -1} n.
$$
Then \emph{wep} GEO-P$(\alpha,\beta,m,p)$ satisfies
$$
N_{\ge k} = \big(1+O(\log^{-1/3} n) \big) \frac {\alpha}{\alpha+1} p^{1/\alpha} n^{(1-\beta)/\alpha} k^{-1/\alpha}.
$$
\end{theorem}

Our next results shows that geo-protean graphs are relatively dense. For a graph $G=(V,E)$ of order $n,$ define the \emph{average degree of $G$} by $d=\frac {2 |E|}{n}.$
\begin{theorem}\label{adeg}
\emph{Wep} the average degree of GEO-P$(\alpha,\beta,m,p)$ is
\begin{equation}
d = (1+o(1)) \frac {p}{1-\alpha} n^{1-\alpha-\beta}.  \label{eq:average1}
\end{equation}
\end{theorem}
Note that the average degree tends to infinity with $n$;  that is, the model generates graphs satisfying a \emph{densification power law}. In~\cite{les1}, densification power laws were reported in several real-world networks such as the physics citation graph and the internet graph at the level of autonomous systems.

\bigskip

Our next result describes the diameter of graphs sampled from the GEO-P model. While the diameter is not shrinking, it can be made constant by allowing the dimension to grow as a logarithmic function of $n.$
\begin{theorem}\label{thm:diam}
Let $\alpha \in (0,1)$, $\beta \in (0,1-\alpha)$, $m \in \N$, and $p\in (0,1]$. Then \emph{wep} the diameter $D$ of GEO-P$(\alpha,\beta,m,p)$ satisfies
\begin{equation} \label{eq:dee}
D = \Omega(n^{\frac {\beta}{(1-\alpha)m}} \log^{ \frac {-\alpha}{m} } n), \mbox{ and } D = O(n^{\frac {\beta}{(1-\alpha)m}} \log^{ \frac {2 \alpha}{(1-\alpha)m} } n).
\end{equation}
In particular, \emph{wep} the order of the diameter can be expressed as:
\[
\log D = \frac {\beta}{(1-\alpha)m} \log n + O \left( \frac{\log\log n}{m} \right).
\]
\end{theorem}
We note that in a geometric model where regions of influence have constant volume and possessing the same average degree as the geo-protean model, the diameter is $\Theta( n^{\frac{\alpha+\beta}{m}} ).$ This is a larger diameter than in the GEO-P model. If $m = C \log n$, for some constant $C>0$, then \emph{wep} we obtain a diameter bounded above by a constant.

\bigskip

Let $G=(V,E)$ be a graph. For sets of vertices $X,Y \subseteq V$, define $e(X,Y)$ to be the set of edges with one endpoint in $X$ and the other in $Y.$ For simplicity, we write $e(X)=e(X,X).$ Let $N(v)$ be the neighbour set of the vertex $v.$ The \emph{clustering coefficient} of vertex $v \in V$ is defined as follows:
$$
c(v) = \frac { e(N(v)) } { {\deg(v) \choose 2} }.
$$
(Note that, formally, we need to assume that $\deg(v) \ge 2$ above, but this is \emph{aas} the case in our model. One can define $c(v)=0$ when $\deg(v) \le 1$.) In other words, $c(v) \in [0,1]$ is the probability that two different neighbours of $v$, chosen uniformly at random, are adjacent. In the random graph $G(n,p)$, the expected value of $c(v)$ is $p$ for any vertex $v$. The \emph{clustering coefficient} of $G$ is defined as
$$
c(G) = \frac {1}{|V|} \sum_{v \in V} c(v).
$$
Hence, $\E( c( G(n,p) )) = p$. We prove that \emph{wep} the GEO-P model, for some values of $m,$ generates graphs with higher clustering coefficient than in a random graph $G(n,d/n)$ with the same expected average degree. It follows from Theorem~\ref{adeg} that
$$
d = (1+o(1)) \frac {p}{1-\alpha} n^{1-\alpha-\beta}.
$$
We use the notation $\lfloor x\rfloor_2$ to denote the largest {\sl even} integer smaller than or equal to $x$.

\begin{theorem}\label{cluster}
\emph{Wep} the clustering coefficient of $G$ sampled from GEO-P$(\alpha,\beta,m,p)$ satisfies the following inequality
\begin{eqnarray*}
c(G) &\ge& (1+o(1)) \left( \frac 34 \left(1 - \frac {2}{3K} \right) \right)^m \left( \frac{1-\alpha }{1+\alpha } \right)p \\
&=& (1+o(1)) \exp \left(  -f\left( \frac mK \right) \right) \left( \frac 34 \right)^m \left( \frac{1-\alpha }{1+\alpha } \right)p,
\end{eqnarray*}
where $f(\frac mK ) = \Theta(\frac mK )$, and
$$
K = \left\lfloor \left( \frac {n^{1-\alpha-\beta} }{\log^3 n} \right)^{1/m} \right\rfloor_2.
$$
\end{theorem}
\noindent Note that if
$$
m \le  (1-\alpha-\beta) \frac {\log n}{\log \log n} \left( 1 - \frac {1}{\log \log n} \right) = (1+o(1)) (1-\alpha-\beta) \frac {\log n}{\log \log n},
$$
then $K \gg m$, and the clustering coefficient of GEO-P$(\alpha,\beta,m,p)$ is \emph{wep} at least
$$
(1+o(1)) \left( \frac 34 \right)^m \left( \frac{1-\alpha }{1+\alpha } \right) p = n^{o(1)} \gg (1+o(1)) \frac {p}{1-\alpha} n^{-\alpha-\beta} = c(G(n,d/n)).
$$
Hence, the clustering coefficient is larger than that of a comparable random graph.

If  $m = o(\log n)$ but large enough so that the condition $K\gg m$ does not hold, a similar result holds but the error term is not $(1+o(1))$ anymore. In this case, \emph{wep},
$$
c(G) \ge \left( \frac 34 \right)^{m + o(m)} = n^{o(1)} \gg c(G(n,d/n)).
$$
Finally, if
$$
m = (1+o(1)) b (1-\alpha-\beta) \log n
$$
for some constant $b \in (0,1)$, then $K \ge \lfloor e^{1/b} \rfloor_2 \ge e^{1/b}-2$, and so \emph{wep}
\begin{eqnarray*}
c(G) &\ge& \left( \frac 34 \left( 1 - \frac {2}{3 (e^{1/b}-2)} \right) (1+o(1)) \right)^m \\
&=& \exp \left( b (1-\alpha-\beta) \log \left( \frac 34 - \frac {1}{2 (e^{1/b}-2)} \right) (1+o(1)) \log n \right).
\end{eqnarray*}
For the random graph counterpart we have that \emph{wep}
$$
c(G(n,d/n)) = (1+o(1)) \frac {p}{1-\alpha} n^{-\alpha-\beta} = \exp \Big( (-\alpha-\beta) (1+o(1)) \log n \Big).
$$
Thus, we get larger clustering coefficient for $b$ small enough; that is, for $b < b_0$, where $b_0 = b_0(\alpha, \beta)$ satisfies the following equation
$$
b (1-\alpha-\beta) \log \left( \frac 34 - \frac {1}{2 (e^{1/b}-2)} \right) = -\alpha-\beta.
$$
(Note that the function on the left hand side is increasing and tends to 0 as $b \to 0$.) For $b > b_0$ we get the opposite behaviour; that is, the clustering coefficient is smaller compared to the binomial random graph counterpart.

\bigskip

The normalized Laplacian of a graph relates to important graph properties; see \cite{sgt}. Let $A$ denote the adjacency matrix and $D$ denote the diagonal degree matrix of a graph $G$.  Then the normalized Laplacian of $G$ is $\mathcal{L} = I - D^{-1/2}AD^{-1/2}.$ Let $0 = \lambda_0 \leq \lambda_1 \leq \cdots \leq \lambda_{n-1} \leq 2$ denote the eigenvalues of $\mathcal{L}$.  The \emph{spectral gap} of the normalized Laplacian is
$$
\lambda = \max\{ |\lambda_1 - 1|, |\lambda_{n-1} - 1| \}.
$$
A spectral gap bounded away from zero is an indication of bad expansion properties. Bad expansion is characteristic for OSNs: see property ($iv$) in the introduction. The next theorem represents a drastic departure from the good expansion found in binomial random graphs, where $\lambda = o(1)$~\cite{sgt,CL}.
\begin{theorem}\label{thm:expa}
Let $\alpha \in (0,1)$, $\beta \in (0,1-\alpha)$, $m \in \N$, and $p\in (0,1]$. Let $\lambda(n)$  be the spectral gap of the normalized Laplacian of GEO-P$(\alpha,\beta,m,p)$. Then \emph{wep}
\begin{enumerate}
  \item If $m = m(n) = o(\log n)$, then $\lambda(n) = 1+o(1).$
  \item If $m = m(n) = C \log n$ for some $C > 0$, then
$$
\lambda (n) \ge 1 - \exp \left( - \frac {\alpha+\beta}{C}\right).
$$
\end{enumerate}
\end{theorem}

\subsection{Dimension of OSNs}

Given an OSN, we describe how we may estimate the corresponding dimension parameter $m$ if we assume the GEO-P model. In particular, if we know the order $n$, power law exponent $b$, average degree $d$, and diameter $D$ of an OSN, then we can calculate $m$ using our theoretical results. Formula (\ref{bee}) gives an estimate for $\alpha$ based on the power law exponent $b$.   If $d^*=\log d/\log n$, then equation (\ref{eq:average1}) implies that, asymptotically, $1-\alpha-\beta = d^*$. If $D^*=\log D/\log n$, then formula (\ref{eq:dee})  about the diameter implies that, asymptotically, $D^*=\frac{\beta}{(1-\alpha)m}$. Thus, an estimate for $m$ is given by:
\begin{equation}
m= \frac {1}{D^*} \left( 1 - \left( \frac{b-1}{b-2} \right) d^*\right) = \frac {\log n}{\log D} \left( 1 - \left( \frac{b-1}{b-2} \right) \frac {\log d}{\log n} \right).  \label{mmm}
\end{equation}

This estimate suggests that the dimension is proportional to  $\log n/\log D$. If $D$ is constant, this means that $m$ grows logarithmically with $n$. Recall that the dimension of an OSN may be roughly defined as the least integer $m$ such that we can accurately embed the OSN in $m$-dimensional Euclidean space. Based on our model we conjecture that the dimension of an OSN is best fit by approximately $\log n.$

The parameters $b,$ $d,$ and $D$ have been determined for samples from OSNs in various studies such as~\cite{ahn, twit0, kwak, mislove2007maa}. The following chart summarizes this data and gives the predicted dimension for each network. We round $m$ up to the nearest integer. Estimates of the total number of users $n$ for Cyworld, Flickr, and Twitter come from Wikipedia~\cite{wiki}, and those from YouTube comes from their website~\cite{yt}.  When the data consisted of directed graphs, we took $b$ to be the power law exponent for the in-degree distribution. As noted in~\cite{ahn}, the power law exponent of $b=5$ for Cyworld holds only for users whose degree is at most approximately $100.$ When taking a sample, we assume that some of the neighbours of each node will be missing. Hence, when computing $d^*,$ we used $n$ equalling the number of users in the sample. As we assume that the diameter of the OSN is constant, we compute $D^*$ with $n$ equalling the total number of users.
\begin{center}
\begin{tabular}{|l||l|l|l|l|}
\hline
\textbf{Parameter} & \textbf{OSN} &  &  &  \\ \hline
& Cyworld & Flickr & Twitter & YouTube \\ \hline\hline
$n$ & $2.4\times 10^{7}$ & $3.2\times 10^{7}$ & $7.5\times 10^{7}$ & $3\times 10^{8}$ \\ \hline
$b$ & $5$ & $2.78$ & $2.4$ & $2.99$ \\ \hline
$d^{\ast }$ & $0.22$ & $0.17$ & $0.17$ & $0.1$ \\ \hline
$D^{\ast }$ & $0.11$ & $0.19$ & $0.1$ & $0.16$ \\ \hline
$m$ & $7$ & $4$ & $5$ & $6$ \\ \hline
\end{tabular}
\end{center}

\section{Proofs of results}\label{proofs}

We will make frequent use of the following standard result about the sum of independent random variables, known as the {\sl Chernoff bound}; for a proof see Theorem~2.8 in \cite{JLR}.

\begin{theorem}\label{eqn:Chernoff}
Let $X$ be a random variable that can be expressed as a sum $X=\sum_{i=1}^n X_i$ of independent random indicator variables where $X_i \in {\rm Be}(p_i)$ with (possibly) different $p_i=\Prob (X_i = 1)= \E X_i$. Then the following holds for $t \ge 0$:
\begin{eqnarray*}
\Prob (X \ge \E X + t) &\le& \exp \left( - \frac {t^2}{2(\E X+t/3)} \right),\\
\Prob (X \le \E X - t) &\le& \exp \left( - \frac {t^2}{2\E X} \right).
\end{eqnarray*}
In particular, if $\eps \le 3/2$, then
\begin{eqnarray*}
\Prob (|X - \E X| \ge \eps \E X) &\le& 2 \exp \left( - \frac {\eps^2 \E X}{3} \right).
\end{eqnarray*}
Moreover, if $\E X \le \log^2 n$, then \emph{wep} $X=O(\log^2 n)$.
\end{theorem}

\smallskip

Before we prove the theorems discussed in the previous sections, we first give a lemma that shows that, if the initial rank is large enough, the rank of a vertex maintains a value close to its initial value until its death. The first lemma is proved in~\cite{JJPP2}.

\begin{lemma}[\cite{JJPP2}] \label{lemma:random_rank}
Suppose that vertex $v$ obtained an initial rank $R \ge \sqrt{n} \log^2 n$ in the ranking by random initial rank scheme. Then \emph{wep}
$$
r(v,t) = R (1+O(\log^{-1/2}n))
$$
to the end of its life.
\end{lemma}

As we will see, the degree of a vertex depends on its rank and its age. The measure used to quantify the age of a vertex is its  {\sl age rank}. The age rank $a(v,t)$ of vertex $v$ at time $t$ is a number between 1 and $n$ which represents the rank of $v$ if all vertices alive at time $t$ are ranked according to age, oldest first. So $a(v,t)=1$ means that $v$ is the oldest vertex alive at time $t$, while $a(v,t)=n$ implies that $v$ is the youngest.

\begin{lemma}\label{lemma:random_rank_bound}
Suppose that vertex $v$ obtained an initial rank $R \ge \log^3 n$. Moreover, $v$ has age rank $a(v,L) \ge Cn$ for some constant $C \in (0,1)$. Then \emph{wep}
$$
r(v,t) \le R \cdot 3^{2\log (1/C)} = O(R)
$$
during the whole process up to time $L$.
\end{lemma}
\begin{proof}
If follows from Theorem~5.5 in~\cite{JJPP2} that \emph{wep} the age rank of a vertex $v$ after
$$
t \le n \log n / 2 - 2 n \log \log n
$$
steps is $(1+o(1)) n \exp(-t/n)$. Since $v$ has age rank at least $Cn$ at time $L$, this implies that $C\leq (1+o(1))\exp(-t_v/n)$, where $t_v$ is the age of $v$ at time $L$ (so $v$ was born at time $L-t_v$). Thus, $t_v\leq (1+o(1))n\log(1/C)$. Recall that the rank of $v$ at the time it is born equals $R$; we wish to find an upper bound on $r(v,L)$ by considering how the rank can change in $t_v$ steps of the process.  Consider the following random variable $X_t$: $X_0 = R$, $X_{t+1} = X_t + 1$ with probability $X_t/n$; $X_{t+1}=X_t$, otherwise.  The variable $X_t$ is an upper bound on the rank of $v$ after $t$ steps. Note that this upper bound only considers changes to the rank due to a vertex of higher rank being inserted, not the change in rank due to vertices of  lower rank being deleted.

We introduce the stopping time:
$$
T = \min \{t \ge 0 : X_t > 3 R\text{ or } t > n/2\}.
$$
For any time $t\leq T$, $X_t$ goes up with probability at most $3 R/ n$. After $n/2$ steps, our random variable increases by at most $(1+o(1)) (3R/n) (n) = (1+o(1)) 3R/2$ \emph{wep}, provided that $n/2 \le T$. Hence, \emph{wep} $T \ge n/2$. Thus, \emph{wep} in the first  $n/2$ steps after its birth, the rank of vertex $v$ can increase by at most a factor of $3$.  Dividing the total life span of $v$ into at most $2\log(1/C)$ blocks of $n/2$ time-steps each, we obtain the result.
\end{proof}

\subsection*{Proof of Theorems~\ref{powerlaw} and \ref{adeg}}

The following theorem shows how the degree of a given vertex depends on its age rank and its initial rank.

\begin{theorem}\label{thm:labels_degree}
Let $\alpha \in (0,1)$, $\beta \in (0,1-\alpha)$, $m \in \N$, $p\in (0,1]$, $i = i(n) \in [n]$. Let $v_i$ be the vertex in GEO-P$(\alpha,\beta,m,p)$ whose age rank at time $L$ equals $a(v_i,L)=i$, and let $R_i$ be the initial rank of $v_i$.

If $R_i \ge \sqrt{n} \log^2 n$, then \emph{wep}
\begin{equation} \label{eq:degrees_bigR}
\deg(v_i,L) = (1+O(\log^{-1/2} n)) p \left( \frac {i}{(1-\alpha)n} + \left(\frac{R_i}{n}\right)^{-\alpha} \frac {n-i}{n} \right) n^{1-\alpha-\beta}.
\end{equation}
Otherwise, that is if $R_i < \sqrt{n} \log^2 n $, \emph{wep}
\begin{equation} \label{eq:degrees_smallR}
\deg(v_i,L) \ge (1+O(\log^{-1/2} n)) p \left( \frac {i}{(1-\alpha)n} + (n^{\alpha/2} \log^{-2 \alpha} n) \frac {n-i}{n} \right) n^{1-\alpha-\beta}.
\end{equation}
\end{theorem}
\begin{proof}
Let $\deg^+(v_i,L)$ denote the number of neighbours of $v_i$ with age rank smaller than $i$, and define $$\deg^-(v_i,L)=\deg(v_i,L)-\deg^+(v_i,L).$$ Let us focus on these random variables independently.

Since vertices are distributed uniformly at random, the expected initial degree of $v_i$ (at the time $v_i$ was born) is
$$
\sum_{r=1}^n p r^{-\alpha} n^{-\beta} = p n^{-\beta} \int_1^n x^{-\alpha} dx + O(1) = \frac {p}{1-\alpha} n^{1-\alpha-\beta} + O(1).
$$
>From the $n-1$ vertices that were older than $v_i$ at the time it was born, only $i-1$ remain. Since vertices are deleted uniformly at random, the expected number of older neighbours remaining equals
$$
\E \deg^+(v_i,L) = \frac {p}{1-\alpha} \ \frac {i}{n} \ n^{1-\alpha-\beta} + O(1).
$$
Since $\deg^+(v_i,L)$ can be expressed as a sum of independent random variables, it follows from Theorem~\ref{eqn:Chernoff} that \emph{wep} this random variable is well concentrated around its expectation , provided that $\E \deg^+(v_i,L) = \Omega(\log^3 n)$. This condition holds if, for example, $i > n^{\alpha+\beta} \log^3 n$. In this case we have that $\deg^+(v_i,L)=(1+O(\log^{-1/2} n)) \E \deg^+(v_i,L)$.

Next we consider the contribution to the degree of $v_i$ of vertices that are younger than $v_i$. Suppose first that $R_i \ge \sqrt{n} \log^2 n$. It follows from Lemma~\ref{lemma:random_rank} that \emph{wep} $r(v_i,t) = R_i (1+O(\log^{-1/2} n))$ to the end of its life. Therefore,
\begin{eqnarray*}
\E \deg^-(v_i,L) &=& (1+O(\log^{-1/2} n)) p R_i^{-\alpha} n^{-\beta} (n-i) \\
&=& (1+O(\log^{-1/2} n)) p \left( \frac {R_i}{n} \right)^{-\alpha} \frac {n-i}{n} n^{1-\alpha-\beta}.
\end{eqnarray*}
Similarly as before, this is well concentrated around its expectation. More precisely, \emph{wep}
$$
\deg^-(v_i,L)=(1+O(\log^{-1/2} n)) \E \deg^-(v_i,L),
$$
provided that $\E \deg^-(v_i,L) = \Omega(\log^3 n)$. Note that it is sufficient to have $i < n - n^{\alpha+\beta} \log^3 n$, even if $R_i=n$.

Let us mention that if one of $d^+(v_i,L)$, $d^-(v_i,L)$ is $O(\log^3 n)$ in expectation, then the other is expected to be $\Omega(n^{1-\alpha-\beta})$ and \emph{wep} is concentrated around its expectation. This implies that \emph{wep} their sum $\deg(v_i,L)$ is always concentrated for $R_i \ge \sqrt{n} \log^2 n$. Combining the number of older and younger neighbours, we obtain that \emph{wep}~(\ref{eq:degrees_bigR}) holds.

Finally, if $R_i < \sqrt{n} \log^2 n$, then \emph{wep} $r(v_i,t) \le \sqrt{n} \log^2 n (1+O(\log^{-1/2} n))$ to the end of its life. The argument for $\deg^+(v_i)$ is analogous and so is omitted, but we only obtain a lower bound for $\deg^-(v_i)$. Thus, we find that \emph{wep}~(\ref{eq:degrees_smallR}) holds.
\end{proof}

\smallskip

It follows from Theorem~\ref{thm:labels_degree} that \emph{wep} the minimum degree is $$(1+o(1)) p n^{1-\alpha-\beta}.$$ Vertices of minimum degree are old and of low rank: they have age rank $i=o(n)$ and thus lost most of their initial links,  and their initial rank $R_i = n - o(n)$, so they never acquired many new links.

The maximum degree is changing during the process and this behaviour is not possible to predict. However, in order to get an upper bound, suppose the extreme case where the oldest vertex is ranked number one during its entire life. In such an extreme case, the degree would be $(1+o(1)) p n^{1-\beta}$ \emph{wep} which indicates that \emph{wep} the maximum degree is at most $(1+o(1)) p n^{1-\beta}$.

\smallskip

We now turn to the average degree.

\begin{proof}[Proof of Theorem~\ref{adeg}] By Theorem~\ref{thm:labels_degree} the average degree is
\begin{eqnarray}
\frac {2 |E|}{n} &=& \frac {2}{n} \sum_{i=1}^n \deg^+(v_i, L) \nonumber\\
&=& (1+o(1)) \ \frac {2}{n} \ \sum_{i=1}^n \frac {p}{1-\alpha} \ \frac{i-1}{n-1} n^{1-\alpha-\beta} \nonumber\\
&=& (1+o(1)) \frac {p}{1-\alpha} n^{1-\alpha-\beta}.  \label{eq:average} \qedhere
\end{eqnarray}
\end{proof}

The proof of Theorem~\ref{powerlaw} is now a simple consequence of Theorem~\ref{thm:labels_degree}. Let $k$ be such that
$$
n^{1-\alpha-\beta} \log^{1/2} n \le k \le n^{1-\alpha/2-\beta} \log^{-2 \alpha -1} n.
$$
One can show that \emph{wep} each vertex $v_i$ that has the initial rank $R_i \ge \sqrt{n} \log^2 n$ such that
$$
\frac {R_i}{n} \ge \big(1+\log^{-1/3} n \big) \left( p n^{1-\alpha-\beta} \frac {n-i}{n} k^{-1} \right)^{1/\alpha}
$$
has fewer than $k$ neighbours, and each vertex $v_i$ for which
$$
\frac {R_i}{n} \le \big(1-\log^{-1/3} n \big) \left( p n^{1-\alpha-\beta} \frac {n-i}{n} k^{-1} \right)^{1/\alpha}
$$
has more than $k$ neighbours.

Let $i_0$ be the largest value of $i$ such that
$$
\left( p n^{1-\alpha-\beta} \frac {n-i}{n} k^{-1} \right)^{1/\alpha} \ge \frac {2 \log^2 n}{\sqrt{n}}.
$$
(Note that $i_0 = n - O(n/\log n)$, since $k \le n^{1-\alpha/2-\beta} \log^{-2 \alpha -1} n$.) Thus,
\begin{eqnarray*}
\E N_{\ge k} &=& \sum_{i=1}^{i_0} \big(1+O(\log^{-1/3} n) \big) \left( p n^{1-\alpha-\beta} \frac {n-i}{n} k^{-1} \right)^{1/\alpha}
+ O \left( \sum_{i=i_0+1}^n \frac {\log^2 n}{\sqrt{n}} \right) \\
&=& \big(1+O(\log^{-1/3} n) \big) \left( p n^{-\alpha-\beta} k^{-1} \right)^{1/\alpha} \sum_{i=1}^{n} (n-i)^{1/\alpha} \\
&=& \big(1+O(\log^{-1/3} n) \big) \left( p n^{-\alpha-\beta} k^{-1} \right)^{1/\alpha} \frac {n^{1+1/\alpha}}{1+1/\alpha} \\
&=& \big(1+O(\log^{-1/3} n) \big) \frac {\alpha}{\alpha+1} p^{1/\alpha} n^{(1-\beta)/\alpha} k^{-1/\alpha}.
\end{eqnarray*}
and the assertion follows from the Chernoff bound, since $k \le n^{1-\alpha/2-\beta} \log^{-2 \alpha -1} n$ and so
$\E Z_{\ge k} = \Omega(\sqrt{n} \log^{2+1/\alpha} n). $ \qed

\subsection*{Proof of Theorem~\ref{thm:diam}}

We consider the upper and lower bounds in separate arguments.

\subsubsection*{Upper bound}

The idea of the proof is to  show that we can construct a connected subgraph of vertices spread out over the hypercube with a small diameter. This subgraph will act as a {\sl backbone}, and the next step in the proof shows that each vertex is connected to the backbone by a path of length at most 2.

We will fix values $A\in (0,1)$ and $R\in [n]$ to suit our needs later.  To construct the backbone we partition the hypercube into $1/A$ subcubes, each of volume $A$. Consider all vertices with initial rank at most $R$ and age rank between $n/4$ and $n/2$; we will call such vertices {\sl eminent} vertices. We now choose $A$ and $R$ such that {\em wep} ($i$) the influence region of each eminent vertex contains the subcube in which it is located as well as all neighbouring subcubes, and ($ii$) each subcube contains at least $\log^2 n$ eminent vertices.

Property ($i$) will be achieved if the sphere of influence of each eminent vertex has size at least $4^m A$ throughout the process. Note that the sphere of influence of an eminent vertex initial has volume at least $R^{-\alpha} n^{-\beta}$, and by Lemmas~\ref{lemma:random_rank} and~\ref{lemma:random_rank_bound} \emph{wep} it remains at least $3^{-2 \alpha \log 4} R^{-\alpha} n^{-\beta} > (R/25)^{-\alpha} n^{-\beta}$ to the end of the process.  We can thus achieve ($i$) by choosing $R$ and $A$ such that the initial influence region is sufficiently larger than $4^m A$ --- we choose $5^m A$. This leads to our first condition on $R$ and $A$:
\begin{equation} \label{eq:first_eq}
(R/25)^{-\alpha}n^{-\beta} = 5^m A.
\end{equation}

It follows from the Chernoff bound that, to guarantee that \emph{wep} every subcube contains at least $\log^2 n$ eminent vertices, it is sufficient to choose $A$ and $R$ so that the expected number of eminent vertices in a subcube is at least $\frac 52 \log^2 n$.  Since the initial rank is independent of age, the expected number of eminent vertices in a subcube equals $(n/4)(R/n)A=(R/4)A$. Thus the following condition on $A$  and $R$ guarantees ($ii$):
\begin{equation} \label{eq:second_eq}
 A \frac {R}{4} = \frac 52 \log^2 n.
\end{equation}

Combining~(\ref{eq:first_eq}) and~(\ref{eq:second_eq}) we find the following values for $R$ and $A$:
\begin{eqnarray*}
R &=& (10 \cdot 25^{-\alpha} \cdot 5^{m}n^\beta \log^2 n)^{1/(1-\alpha)}\\
\frac {1}{A} &=& \frac {(10\cdot 25^{-\alpha})^{1/(1-\alpha)}}{10} 5^{\frac {m}{1-\alpha}} n^{\frac{\beta}{1-\alpha}} \log^{\frac {2\alpha}{1-\alpha}} n.
\end{eqnarray*}

The set of all eminent vertices form the \emph{backbone}. Since vertices in each subcube induce a random graph with a parameter $p$ (constant, bounded away from zero), \emph{wep} the induced subgraph is connected and of diameter of two. For any two neighbouring subcubes, the probability that there is no edge between them is equal to
$$
(1-p)^{{\log^2 n \choose 2}} = \exp ( -  \Omega(\log^4 n) ),
$$
so \emph{wep} there is at least on edge connecting two neighbouring subcubes. Since the largest metric distance in the hypercube with torus metric equals 1/2, and the subcubes have diameter $2A^{1/m}$, the diameter of the backbone is \emph{wep}
$$
O \left( (1/A)^{1/m} \right) = O( n^{\frac{\beta}{(1-\alpha)m}} \log^{\frac {2\alpha}{(1-\alpha)m}} n )
$$

To finish the proof, we will show that \emph{wep} vertex $v$ that is not in the backbone is within distance two from some vertex in the backbone. Consider any vertex $v$ not part of the backbone. Consider $S_v$, the ball of volume $n^{-\alpha-\beta}$ centered at $v$. The volume of $S_v$ is the minimum volume of a sphere of influence, so for each vertex $w$ in $S_v$, edge $vw$ exists with probability $p$. Note also that every vertex of age rank greater than $n/2$ in any subcube links to the backbone vertex in the same subcube with probability $p$, since the sphere of influence of the backbone vertex includes all of the subcube. There are $(1+o(1))n^{1-\alpha-\beta}/2$ vertices of age rank greater than $n/2$ in $S_v$, and a path of length 2 from $v$ to the backbone using that vertex exists with probability $p^2$. Thus, {\sl wep} such a path exists. \qed

\subsubsection*{Lower bound}

Note that for $m = \Theta(\log n)$ the theorem states that $D \ge 1$ and the lower bound trivially holds. We can assume then that $m = o(\log n)$. Let $R = n^{ \frac {\beta}{1-\alpha} } \log^{-1} n$. Note that at every point of the process, the union of influence regions of vertices with rank at most $R$ has volume at most
$$
\sum_{r=1}^R r^{-\alpha} n^{-\beta} = \Theta(R^{1-\alpha} n^{-\beta}) = o(1).
$$
These vertices can generate \emph{long} edges, we will call such vertices {\sl hubs}; other edges are \emph{short}. The length of every short edge is \emph{wep} at most
$$
(1+o(1)) \left( R^{-\alpha} n^{-\beta} \right)^{1/m} = (1+o(1)) n^{ \frac{-\beta}{(1-\alpha)m}} \log^{\frac {\alpha}{m}} n,
$$
since the rank at least $R$ is well concentrated. (This time the length corresponds to the torus metric, not the graph distance.)

Since the total volume of the influence regions of the hubs is $o(1)$, there must be a vertex $v$ and a constant $c\in(0,1/3)$ so that \emph{wep} the ball around $v$ with radius $c$ does not intersect any hub of an influence region. Moreover, since vertices are uniformly distributed, \emph{wep} there is a vertex $u$ at (metric) distance greater than $c$ from $v$. Any path from $v$ to $u$ must use short edges to bridge the distance from $v$ to the edge of the circle with radius $c$. This implies that \emph{wep} the path has (graph distance) length at least $(1+o(1)) n^{ \frac{\beta}{(1-\alpha)m}} \log^{\frac {-\alpha}{m}} n$, which finishes the proof. \qed

\subsection*{Proof of Theorem~\ref{cluster}}

Consider $G$ sampled from GEO-P$(\alpha,\beta,m,p)$, and define $V_{min} = n^{-\alpha - \beta}.$ Then in the GEO-P$(\alpha,\beta,m,p)$, $V_{min}$ corresponds to the lower bound for a volume of the influence region that is obtained for a vertex with rank $n.$ Thus, if the distance between $u$ and $v$ is at most $r_{min}=V_{min}^{1/m}/2$,  two vertices $u$ and $v$ are adjacent with probability $p$ . (Of course, edges can also be created between vertices which are a larger distance apart, but this requires additional conditions on the initial rank of these vertices.)  Recall that, due to our choice of metric, the ball of radius $r_{min}$ centered at $v$ is a hypercube. We will call this the {\sl minimal hypercube} of $v$. The minimal region of a vertex is always a subset of its region of influence.

Fix $v \in V(G)$. Without loss of generality, we can assume that $v$ is the origin of the hypercube $S$; that is, $v = (0,0,\dots, 0) \in S$. In order to estimate $c(v)$ from below we partition the minimal hypercube of $v$ into  $K^m$  disjoint, identical, small hypercubes of volume $\log^3 n / n$ each. The expected number of neighbours of $v$ in every ball is $p \log^3 n$, so \wep the number of neighbours equals $(1+O(\log^{-1/2}n)) p \log^3 n$.
Note that
\begin{eqnarray*}
K &=& \left\lfloor \left( \frac {V_{min} n }{\log^3 n} \right)^{1/m} \right\rfloor_2 \\
&=& \left\lfloor \left( \frac {n^{1-\alpha-\beta} }{\log^3 n} \right)^{1/m} \right\rfloor_2 \\
&=& \left( \frac {n^{1-\alpha-\beta} }{\log^3 n} \right)^{1/m} (1+O(1/K)).
\end{eqnarray*}
Recall that $\lfloor x\rfloor_2$ is the largest even integer smaller than or equal to $x$.

We index the balls as follows:
\begin{eqnarray*}
b_{i_1, i_2, \dots,i_m} &=& (a_{i_1-1},a_i)\times \dots \times (a_{i_m-1},a_i),
\end{eqnarray*}
where
$a_i = - \frac {V_{min}^{1/m}}{2} + i \left( \frac {\log^3 n}{n} \right)^{1/m}  $.
For all $s \in [m]$, $i_s$ takes values $i_s = 1, 2, \dots, K$.

Note that every vertex in $ b_{i_1, i_2, \dots,i_m}$ falls into the influence region of every vertex in $b_{j_1, j_2, \dots,j_m}$ if the following condition holds:
\begin{description}
\item[$(\ast )$] For all $s \in [m]$, $|i_s-j_s| \le K/2-1$.
\end{description}
Using this observation, we can derive a lower bound on  $e(N(v))$, the number of edges in the neighbourhood of $v$.  By symmetry, we have that \emph{wep}
\begin{eqnarray*}
e(N(v)) &\ge& \frac {1}{2} \sum_{i_1=1}^{K} \dots \sum_{i_m=1}^{K}\sum_{j_1 = 1, (*)}^K \dots \sum_{j_m = 1, (*)}^K e( b_{i_1, i_2, \dots,i_m}, b_{j_1, j_2, \dots,j_m})  \\
&=& \frac {1}{2} \sum_{i_1=1}^{K} \dots \sum_{i_m=1}^{K}\sum_{j_1 = 1, (*)}^K \dots \sum_{j_m = 1, (*)}^K  (1+o(1)) (p \log^3 n)^2 p \\
&\ge& (1+o(1)) \frac {2^m}{2}  \sum_{i_1=1}^{K/2} \dots \sum_{i_m=1}^{K/2}\sum_{j_1 = 1}^{K/2-1+i_1} \dots \sum_{j_m = 1}^{K/2-1+i_m} (p \log^3 n)^2 p ;
\end{eqnarray*}
we use the notation $\sum_{j_s = 1, (*)}^K$ to indicate that the sum is over all $j_s$ between 1 and $n$ such that $(*)$ holds.

By symmetry, we have that the case where $i_s\leq K/2$ is symmetrical to the case where $K- i_s +1 \leq K/2$. So we can count $e(N(v))$ by considering only the hypercubes $b_{i_1,\dots ,i_m}$ where $i_s<K/2$ for $1\leq s\leq m$, and multiplying the result by $2^m$. Using this symmetry argument, we have that \emph{wep}
\begin{eqnarray*}
e(N(v)) &\ge&
 \frac {2^m}{2}  \sum_{i_1=1}^{K/2} \dots \sum_{i_m=1}^{K/2}\sum_{j_1 = 1}^{K/2-1+i_1} \dots \sum_{j_m = 1}^{K/2-1+i_m}
 e( b_{i_1, i_2, \dots,i_m}, b_{j_1, j_2, \dots,j_m})  \\
& = & (1+o(1)) \frac {2^m}{2}  \sum_{i_1=1}^{K/2} \dots \sum_{i_m=1}^{K/2}\sum_{j_1 = 1}^{K/2-1+i_1} \dots \sum_{j_m = 1}^{K/2-1+i_m} (p \log^3 n)^2 p ;
\end{eqnarray*}
we use the notation $\sum_{j_s = 1, (*)}^K$ to indicate that the sum is over all $j_s$ between 1 and $n$ such that $(*)$ holds.

Thus, \emph{wep}
\begin{eqnarray}
e(N(v)) &\ge& (1+o(1)) \frac {2^m}{2} \sum_{i_1=1}^{K/2} \dots \sum_{i_m=1}^{K/2} \left( \frac K2-1+i_1 \right) \dots \left( \frac K2-1+i_m \right) (p \log^3 n)^2 p \nonumber \\
&=& (1+o(1)) \frac {2^m}{2}\left( \left(\frac {K/2 + (K-1)}{2}\right) \frac {K}{2} \right)^m (p \log^3 n)^2 p \nonumber\\
&=& (1+o(1)) \frac {1}{2} \left( \frac 34 \left(1 - \frac {2}{3K} \right) K^2 \right)^m (p \log^3 n)^2 p \nonumber\\
&=& (1+o(1)) \exp \left(- O \left( \frac mK \right) \right) \left( \frac 34 \right)^m   \frac {\Big(p n^{1-\alpha-\beta} \Big)^2}{2}p \label{lb}.
\end{eqnarray}

Fix a vertex $v_i \in V(G)$, $i = an$ for some $a \in [0,1]$. Suppose that the initial rank of $v_i$ is $R_i = bn$ for some $b \in (0,1]$. It follows from Theorem~\ref{thm:labels_degree} that \emph{wep}
$$
\deg (v_i) = (1+o(1)) p \left( \frac {a}{1-\alpha} + b^{-\alpha} (1-a) \right) n^{1-\alpha-\beta}.
$$
Hence, by (\ref{lb}) \emph{wep} $c(v_i)$ may be bounded from below as follows:
\begin{eqnarray*}
c(v_i) &\ge& (1+o(1)) \exp \left(O \left( \frac mK \right) \right) \left( \frac 34 \right)^m   \frac { (pn^{1-\alpha-\beta})^2 / 2 } { {\deg(v_i,L) \choose 2} } p\\
&=&  (1+o(1)) \exp \left(O \left( \frac mK \right) \right) \left( \frac 34 \right)^m  \left( \frac {a}{1-\alpha} + b^{-\alpha} (1-a) \right)^{-2}p.
\end{eqnarray*}
Since the initial rank is distributed uniformly at random, we derive that the expected value of $c(v_i)$ can be bounded as follows:
$$
\E (c(v_i)) \ge (1+o(1)) \exp \left(O \left( \frac mK \right) \right) \left( \frac 34 \right)^m p \int_0^1 \left( \frac {a}{1-\alpha} + b^{-\alpha} (1-a) \right)^{-2} db.
$$
Therefore, we find that \emph{wep}
$$
c(G) \ge (1+o(1)) \exp \left(O \left( \frac mK \right) \right) \left( \frac 34 \right)^m p D(\alpha),
$$
where
$$
D(\alpha) = \int_0^1 \int_0^1 \left( \frac {a}{1-\alpha} + b^{-\alpha} (1-a) \right)^{-2} db \ da.
$$
After changing the order of integration, it is straightforward to see that
\begin{eqnarray*}
D(\alpha) &=&
\int_0^1 \int_0^1 \left( a \left( \frac {1}{1-\alpha} - b^{-\alpha} \right) + b^{-\alpha} \right)^{-2} da \ db \\
&=& - \int_0^1 \left[ \left( a \left( \frac {1}{1-\alpha} - b^{-\alpha}
\right) + b^{-\alpha} \right)^{-1} \left( \frac {1}{ \frac {1}{1-\alpha}
- b^{-\alpha} }\right) \right]_{a=0}^{a=1} \\
&=& - \int_0^1 \left( (1-\alpha) - b^{\alpha} \right) \left( \frac {1}{ \frac {1}{1-\alpha} - b^{-\alpha} } \right) \ db \\ &=& \left[ \frac {1-\alpha}{1+\alpha} \ b^{1-\alpha} \right]_{b=0}^{b=1} =
\frac {1-\alpha}{1+\alpha},
\end{eqnarray*}
which finishes the proof of the theorem. \qed

\smallskip

This proof shows that, as the dimension $m$ increases, this lower bound on the  size of $e(N(v))$, and thus the number of triangles, decreases exponentially. This concurs with our interpretation of the dimension as the minimal number of attributes that characterizes a user in a social network. It makes sense that assortativity decreases as the dimension increases: if a user has a friend $A$ with whom she shares interests in one dimension, and a friend $B$ with whom she connects in another dimensions, then the chances that $A$ and $B$ become friends may not be much larger than dictated by random chance.

\subsection*{Proof of Theorem~\ref{thm:expa}}

In order to prove Theorem~\ref{thm:expa} we show that there are sparse cuts in our model. As in the previous subsection, for sets $X$ and $Y$ we use the notation $e(X,Y)$ for the number of edges with one end in each of $X$ and $Y$. Suppose that the unit hypercube $S=[0,1]^m$ is partitioned into two sets of the same volume,
$$
S_1 = \{ x=(x_1, x_2, \dots, x_m) \in S : x_1 \le 1/2 \},
$$
and $S_2 = S \setminus S_1$. Both $S_1$ and $S_2$ contain $(1+o(1)) n/2$ vertices \emph{wep}. In a good expander (for instance, the binomial random graph $G(n,p)$), \emph{wep} there would be
$$
(1+o(1)) \frac {|E|}{2} = (1+o(1)) \frac{p}{4(1-\alpha)} n^{2-\alpha-\beta}
$$
edges between $S_1$ and $S_2$. Below we show that it is not the case in our model.

\begin{theorem}\label{thm:cuts}
Let $\alpha \in (0,1)$, $\beta \in (0,1-\alpha)$, $m \in \N$, and $p\in (0,1]$. Then \emph{wep} GEO-P$(\alpha,\beta,m,p)$ has the following property:
\begin{enumerate}
  \item if $m = m(n) = o(\log n)$, then $e(S_1,S_2) = o(n^{2-\alpha-\beta})$,
  \item if $m = m(n) = C \log n$ for some $C > 0$, then
$$
e(S_1,S_2) \le (1+o(1)) \frac{p}{4(1-\alpha)} n^{2-\alpha-\beta} \exp \left( - \frac {\alpha+\beta}{C}\right).
$$
\end{enumerate}
\end{theorem}
\begin{proof}
Let us call a vertex $v$ \emph{dangerous} if, at some point of the protean process, the influence region of $v$ has nonempty intersection with both $S_1$ and $S_2$; that is, there exists $t$ such that $R(v,t) \cap S_1 \neq \emptyset$ and $R(v,t) \cap S_2 \neq \emptyset$. It is easy to see that the older vertex of every edge between $S_1$ and $S_2$ must be dangerous. We are going to estimate $e(S_1,S_2)$ by investigating the number of dangerous vertices with the final rank from a given range. Since the number of edges adjacent to a given dangerous vertex in the cut can be estimated by its degree, the conclusion can be obtained.

First, let us consider vertices with small ranks; that is, vertices with $r(v,L) < \sqrt{n}\log^2 n$. Unfortunately, for these vertices we cannot control the behaviour of the corresponding random variables $r(v,t)$ during the process. However, deterministically $|R(v,t)|=O(n^{-\beta})$ for all vertices at any point of the process. Since vertices are distributed in $S$ uniformly at random, the expected number of dangerous vertices with $r(v,L) < \sqrt{n}\log^2 n$ can be estimated by $O(n^{-\beta/m}) \sqrt{n} \log^2 n = O(n^{1/2-\beta/m} \log^2 n)$. Hence, \emph{wep} the number of dangerous vertices from this range is $$O(n^{\max\{1/2-\beta/m,0\}} \log^2 n)$$ by the Chernoff bound. As we already mentioned, it is not possible to estimate the number of edges adjacent to these vertices. However, by considering the extreme case, that is, when these vertices have the smallest possible ranks during the whole time, we obtain that \emph{wep} the number of edges between these vertices and older ones is at most
\begin{eqnarray*}
\sum_{r=1}^{O(n^{\max\{1/2-\beta/m,0\}} \log^2 n)} r^{-\alpha} n^{1-\beta} &=& O(n^{1-\beta+(1-\alpha) \max\{1/2-\beta/m,0\}} \log^{2(1-\alpha)} n) \\
&=& o(n^{2-\alpha-\beta}).
\end{eqnarray*}

Now, for a given $i = 1, 2, \dots, \frac 12 \log_2 n - 2 \log_2 \log n$, consider vertices with
\begin{equation} \label{eq:range}
2^{i-1} \sqrt{n} \log^2 n \le r(v,L) < 2^{i} \sqrt{n} \log^2 n.
\end{equation}

By Lemma~\ref{lemma:random_rank}, \emph{wep} $r(v,t) = (1+O(\log^{-1/2} n)) r(v,L)$ (which implies that $|R(v,t)| = (1+O(\log^{-1/2} n)) r(v,L)^{-\alpha} n^{-\beta}$) during the whole its life. It follows from the Chernoff bound that \emph{wep} the number of dangerous vertices that satisfy~(\ref{eq:range}) is at most $$O((2^i \sqrt{n} \log^2 n)^{1-\alpha/m} n^{-\beta/m})$$ and the number of edges adjacent to these vertices can be estimated by
\begin{align*}
O((2^i \sqrt{n} \log^2 n)^{1-\alpha/m} & n^{-\beta/m}) \cdot O( (2^i \sqrt{n} \log^2 n)^{-\alpha} n^{1-\beta})  \\
& = O((2^i \sqrt{n} \log^2 n)^{1-\frac{m+1}{m} \alpha} n^{1-\frac{m+1}{m}\beta}).
\end{align*}
Finally, \emph{wep} the number of edges in the cut, that are adjacent to vertices with final ranks at least $\sqrt{n} \log^2 n$, is
\begin{align*}
\sum_{i=1}^{\frac 12 \log_2 n - 2 \log_2 \log n} &O((2^i \sqrt{n} \log^2 n)^{1-\frac{m+1}{m} \alpha} n^{1-\frac{m+1}{m}\beta}) \\
&= \left\{
    \begin{array}{ll}
      O(n^{2-\frac{m+1}{m} \alpha-\frac{m+1}{m} \beta}) , & \hbox{if $\alpha < \frac {m}{m+1}$;} \\
      O(n^{1-\frac{m+1}{m} \beta} \log n), & \hbox{if $\alpha = \frac {m}{m+1}$;} \\
      O(n^{\frac 32 - \frac{m+1}{m} \frac {\alpha}{2} -\frac{m+1}{m} \beta} \log^{2 - 2 \frac {m+1}{m} \alpha} n), & \hbox{if $\alpha > \frac {m}{m+1}$,}
    \end{array}
  \right.
\end{align*}
which is $o(n^{2-\alpha-\beta})$, provided that $m = o(\log n)$. Hence, item (1) of the theorem follows.

For item (2) in the case $m = C \log n$ for some $C>0$, we must be more precise and take care of constants hidden in the $O(\cdot)$ notation. It can be shown that \emph{wep}
\begin{eqnarray*}
e(S_1,S_2) &\le& (1+o(1)) \frac{p}{4(1-\alpha)} n^{2-\frac{m+1}{m} \alpha-\frac{m+1}{m} \beta} \\
&=& (1+o(1)) \frac{p}{4(1-\alpha)} n^{2-\alpha-\beta} \exp \left( - (\alpha+\beta) \frac {\log n}{m} \right) \\
&=& (1+o(1)) \frac{p}{4(1-\alpha)} n^{2-\alpha-\beta} \exp \left( - \frac {\alpha+\beta}{C}\right),
\end{eqnarray*}
and the proof is complete.
\end{proof}

\bigskip

To finish the proof of Theorem~\ref{thm:expa}, we use the expander mixing lemma for the normalized Laplacian (see~\cite{sgt} for its proof). For sets of nodes $X$ and $Y$ we use the notation $\vol(X)$ for the volume of the subgraph induced by $X$, $\bar{X}$ for the complement of $X$, and, as introduced before, $e(X,Y)$ for the number of edges with one end in each of $X$ and $Y.$ (Note that $X \cap Y$ does not have to be empty; in general, $e(X,Y)$ is defined to be the number of edges between $X \setminus Y$ to $Y$ plus twice the number of edges that contain only vertices of $X \cap Y$.)
\begin{lemma}\label{mix}
For all sets $X \subseteq G,$
\[
\left| e(X,X) - \frac{(\vol(X))^{2}}{\vol(G)} \right| \leq \lambda \frac{\vol(X)\vol(\bar{X})}{\vol(G)}.
\]
\end{lemma}

\bigskip

\begin{proof}[Proof of Theorem~\ref{thm:expa}]
It follows from~(\ref{eq:average}) and Chernoff bound that \emph{wep}
\begin{eqnarray*}
\vol (G_L) &=& (1+o(1)) \frac {p}{1-\alpha} n^{2-\alpha-\beta}  \\
\vol (S_1) &=& (1+o(1)) \frac {p}{2(1-\alpha)} n^{2-\alpha-\beta} ~~=~~ \vol (S_2).
\end{eqnarray*}

Suppose first that $m=o(\log n)$. From Theorem~\ref{thm:cuts} we derive that \emph{wep}
\begin{eqnarray*}
e(S_1,S_1) &=& \vol (S_1)-e(S_1,S_2) ~~=~~ (1+o(1)) \vol (S_1) \\
&=& (1+o(1)) \frac {p}{2(1-\alpha)} n^{2-\alpha-\beta},
\end{eqnarray*}
and Lemma~\ref{mix} implies that \emph{wep} $\lambda_n \ge 1+o(1)$. By definition, $\lambda_n \le 1$ so $\lambda_n = 1+o(1)$.

Suppose now that $m=C \log n$ for some constant $C>0$. Using Theorem~\ref{thm:cuts} one more time, we find that \emph{wep}
$$
e(S_1,S_1) = (1+o(1)) \frac {p}{1-\alpha} n^{2-\alpha-\beta} \left( \frac 12 - \frac {\exp \left( - \frac {\alpha+\beta}{C}\right)}{4} \right).
$$
As before, the assertion follows directly from Lemma~\ref{mix}.
\end{proof}

\section{Conclusion and Discussion}\label{conc}

We introduced the geo-protean (GEO-P) geometric model for OSNs, and showed that with high probability, the model generates graphs satisfying each of the properties (i) to (iv) in the introduction. We introduce the dimension of an OSN based on our model, and examine this new parameter using actual OSN data. We observed that the dimension of various OSNs ranges from four to $7.$ It may therefore, be possible to group users via a relatively small number of attributes, although this remains unproven. The \emph{Logarithmic Dimension Hypothesis} (or \emph{LDH}) conjectures that  the dimension of an OSN is best fit by $\log n$, where $n$ is the number of users in the OSN.

The ideas of using geometry and dimension to explore OSNs deserves to be more thoroughly investigated. Given the availability of OSN data, it may be possible to fit the data to the model to determine the dimension of a given OSN. Initial estimates from actual OSN data indicate that the spectral gap found in OSNs correlates with the spectral gap found in the GEO-P model when the dimension is approximately $\log n$, giving some credence to the LDH. Another interesting direction would be to generalize the GEO-P to a wider array of ranking schemes (such as ranking by age or degree), and determine when similar properties (such as power laws and bad spectral expansion) provably hold.

We finish by mentioning that recent work~\cite{cklpr} indicates that social networks lack high compressibility, especially in contrast to the web graph. We propose to study the relationship between the GEO-P model and the incompressibility of OSNs in future work.

\end{document}